\pdfoutput=1
\RequirePackage{ifpdf}
\ifpdf 
\documentclass[pdftex]{sigma}
\else
\documentclass{sigma}
\fi

\numberwithin{equation}{section}

\newtheorem{thm}{Theorem}[section]
\newtheorem{prop}[thm]{Proposition}
\newtheorem{cor}[thm]{Corollary}
{ \theoremstyle{definition}
\newtheorem{rmk}[thm]{Remark}}

\newcommand{\Tr}{\mathop{\rm Tr}}

\newcommand{\Det}{\mathop{\rm Det}}

\newcommand{\cV}{\mathcal{V}}
\newcommand{\cL}{\mathcal{L}}

\newcommand{\bu}{\bar u}
\newcommand{\bv}{\bar v}

\newcommand{\lam}{\lambda}

\newcommand\cA{{\mathcal A}}

\newcommand\cK{{\mathcal K}}
\newcommand\cH{{\mathcal H}}

\newcommand{\kb}{\tilde \kappa}
\newcommand{\ka}{ \kappa}
\newcommand{\kp}{ \kappa^+}
\newcommand{\km}{ \kappa^-}
\newcommand{\kab}{\tilde \kappa}

\newcommand\CC{\mathbb C}

\newcommand\II{\mathbb I}

\def\<{\langle}
\def\>{\rangle}
\def\rvec{|0\>}
\def\rvech{|\hat 0\>}

\def\bbb{\mathbb{B}}

\begin{document}
\allowdisplaybreaks

\newcommand{\arXivNumber}{1804.00597}

\renewcommand{\PaperNumber}{054}

\FirstPageHeading

\ShortArticleName{Modified Algebraic Bethe Ansatz: Twisted XXX Case}
\ArticleName{Modified Algebraic Bethe Ansatz: Twisted XXX Case}

\Author{Samuel BELLIARD~$^{\dag\ddag}$, Nikita A.~SLAVNOV~$^\S$ and Benoit VALLET~$^\ddag$}

\AuthorNameForHeading{S.~Belliard, N.A.~Slavnov and B.~Vallet}

\Address{$^\dag$~Sorbonne Universit\'e, CNRS, Laboratoire de Physique Th\'eorique et Hautes Energies,\\
\hphantom{$^\dag$}~LPTHE, F-75005, Paris, France}
\EmailD{\href{mailto:samuel.belliard@gmail.com}{samuel.belliard@gmail.com}}

\Address{$^\ddag$~Institut de Physique Th\'eorique, DSM, CEA, URA2306 CNRS Saclay,\\
\hphantom{$^\ddag$}~F-91191, Gif-sur-Yvette, France}
\EmailD{\href{mailto:benoit.vallet@u-psud.fr}{benoit.vallet@u-psud.fr}}

\Address{$^\S$~Steklov Mathematical Institute of Russian Academy of Sciences, Moscow, Russia}
\EmailD{\href{mailto:nslavnov@mi.ras.ru}{nslavnov@mi.ras.ru}}

\ArticleDates{Received April 09, 2018, in final form May 28, 2018; Published online June 07, 2018}

\Abstract{We prove the modified algebraic Bethe Ansatz characterization of the spectral problem for the closed XXX Heisenberg spin chain with an arbitrary twist and arbitrary positive (half)-integer spin at each site of the chain. We provide two basis to characterize the spectral problem and two families of inhomogeneous Baxter T-Q equations. The two families satisfy an inhomogeneous quantum Wronskian equation.}


\Classification{82R53; 81R12}

\vspace{-1mm}

\section{Introduction}

The question of the spectral problem of the lattice quantum integrable models without ${\rm U}(1)$ symmetry, as the XXZ Heisenberg spin chain on the segment with the most general integrable boundary condition \cite{Skl88}, has led to the development of new techniques to perform the Bethe ansatz. Among the proposed methods such as the {\it off diagonal Bethe ansatz} \cite{CYSW13c} or the {\it separation of the variables} \cite{Der, Skl85}, the {\it modified algebraic Bethe ansatz} (MABA) proposes to understand this new method from the algebraic Bethe ansatz point of view. In this paper we study this method and determine what the changes are with respect to the case with ${\rm U}(1)$ symmetry and what has to be kept.

The new features in the Bethe ansatz for models without ${\rm U}(1)$ symmetry are as follows. Firstly, the Baxter T-Q equation has a new term~\cite{CYSW13a,CYSW13b}. Secondly, the Bethe vectors are linear combinations of all Bethe vectors of the associated model where the ${\rm U}(1)$ symmetry is restored, and these vectors are factorized in terms of a modified creation operators~\cite{Bel14,BP152, BC13}, that was proved in \cite{ZLCYSW15b, ZLCYSW15a} by means of the Baxter T-Q equation\footnote{Contrary to this method, the MABA provides a constructive way to build Bethe vectors and clarifies the algebraic origin of the new term in the Baxter T-Q equation.}. Thirdly, the number of Bethe roots is fixed and depends on the model under consideration~\cite{CYSW13b, Nep}.

The MABA also can be applied to the models with quasi-periodic boundary conditions such as the XXX spin Heisenberg chain with an arbitrary twist. In particular, one can use this approach for studying XXX spin-$\tfrac 12$ chain
\begin{gather} \label{HXXX}
H= \sum_{k=1}^N\big(\sigma^{x}_k\otimes\sigma^{x}_{k+1}+\sigma^{y}_k\otimes\sigma^{y}_{k+1}+\sigma^{z}_k\otimes\sigma^{z}_{k+1}\big),
\end{gather}
subject to the following non-diagonal boundary conditions:
\begin{gather*}
 \gamma \sigma^{x}_{N+1}= \frac{ \kab^2+ \kappa^2-\kappa_+^2-\kappa_-^2}{2}\sigma^{x}_1 +i\frac{\kappa^2-\kab^2-\kappa_+^2+\kappa_-^2}{2}\sigma^{y}_1 +(\kappa \kappa_- -\kab \kappa_+)\sigma^{z}_1, \\
 \gamma \sigma^{y}_{N+1}= i\frac{\kab^2- \kappa^2-\kappa_+^2+\kappa_-^2}{2}\sigma^{x}_1 +\frac{\kab^2+\kappa^2+\kappa_+^2+\kappa_-^2}{2}\sigma^{y}_1 -i(\tilde \kappa \kappa_++\kappa \kappa_-)\sigma^{z}_1, \\
 \gamma \sigma^{z}_{N+1}=(\kappa \kappa_+-\tilde \kappa \kappa_-)\sigma^{x}_1 +i(\tilde \kappa \kappa_-+\kappa \kappa_+)\sigma^{y}_1 +(\tilde \kappa \kappa+\kappa_+ \kappa_-)\sigma^{z}_1.
\end{gather*}
The twist parameters $\{\kappa,\tilde \kappa,\kappa_+,\kappa_-\} \in \CC^4$ are generic and $\gamma =\kab\kappa -\kappa_+\kappa_-$. The Pauli matrices\footnote{$\sigma^{z}=\left(\begin{smallmatrix}
 1 & 0\\
 0 & -1 \end{smallmatrix}\right)$, $\sigma^{+}=\left(\begin{smallmatrix}
 0 & 1\\
 0 & 0 \end{smallmatrix}\right)$, $\sigma^{-}=\left(\begin{smallmatrix}
 0 & 0\\
 1 & 0 \end{smallmatrix}\right)$, $\sigma^{x}=\sigma^{+}+\sigma^{-}$, $\sigma^{y}=i(\sigma^{-}-\sigma^{+}) $.}~$\sigma^{\alpha}_k$ with $\alpha=x,y,z$ act non-trivially on the $k$th component of the quantum space $\cH= \otimes_{k=1}^N V_k$ with $V_k=\CC^2$. This model provides a simple example where the new properties of the method can be studied. In particular, in the case of the Hamiltonian~\eqref{HXXX}, several conjectures about the MABA were formulated in~\cite{BP15}. There, a special transformation of the twist matrix was proposed. One more conjecture of~\cite{BP15}, formulated on the base of direct calculations for the chains of small length, concerns a special off-shell action of the modified creation operator, which generates a~new term of the Baxter T-Q equation.

In this paper we consider the most general case of the closed XXX Heisenberg spin chain with an arbitrary twist and arbitrary positive (half)-integer spins transfer matrix\footnote{Explicit Hamiltonian for this case is not known. It can be formulated in the more restrictive case of $L_0$-regular spin chains~\cite{BCR09}, where $L_0$ arbitrary positive (half)-integer spins are periodically repeated. The case $L_0=1$ with spin $\frac12$ corresponds to the Hamiltonian \eqref{HXXX}.}. We prove the conjecture about the multiple action of the modified creation operator on the pseudo-vacuum state given in~\cite{BP15}. In fact, we prove this property independently of the action on the pseudo-vacuum state and go beyond the proofs done for other models \cite{ABGP15,Cram14,Cram17}. Moreover, we consider two different bases for solving the spectral problem of this family of models and we relate the two solutions by a modified quantum Wronskian equation.

This paper is organized as follows. In Section~\ref{S:RTT} we introduce a relevant notation and recall the main tools to implement the ABA: the Yangian of $\mathfrak{gl}_2$ and the associated representation to define arbitrary positive (half)-integer spins transfer matrix. In Section~\ref{S:ABA} we recall, for the case with diagonal twist ($\kp=\km=0$), the ABA scheme and the associated quantum Wronskian equation. Section~\ref{S:MABA} is devoted to the case with non-diagonal twist ($\kp\neq 0$ and $\km\neq 0$) and to the study of the MABA scheme. Here we give two basis for the Bethe vectors as well as their corresponding spectrum, which allows us to obtain the modified quantum Wronskian equation. Concluding remarks are given in Section~\ref{Conc}.

\section[The Yangian of $\mathfrak{gl}_2$ and Bethe subalgebra]{The Yangian of $\boldsymbol{\mathfrak{gl}_2}$ and Bethe subalgebra}\label{S:RTT}
Let us recall the simplest rational $R$-matrix, which is an element of $\operatorname{End}\big(\CC^2\otimes \CC^2\big)$:
 \begin{gather*} \label{Rm}
 R(u)=\frac{u}{c} I+ P.
 \end{gather*}
Here $c$ is a constant, $I=\sum\limits_{i,j=1}^2 E_{ii}\otimes E_{jj}$ is the identity operator on $\CC^2\otimes \CC^2$, $P=\sum\limits_{i,j=1}^2 E_{ij}\otimes E_{ji}$ is the permutation operator on $\CC^2\otimes \CC^2$, and $(E_{ij})_{kl}=\delta_{ik}\delta_{jl}$. This $R$-matrix is a~solution of the Yang--Baxter equation
 \begin{gather*}
 R_{ab}(u-v) R_{ac}(u-w) R_{bc}(v-w)= R_{bc}(v-w) R_{ac}(u-w) R_{ab}(u-v).
 \end{gather*}
Here the $R$-matrix is understood as an element of $\operatorname{End}(V_a \otimes V_b \otimes V_c)$ with $V=\CC^2$. Namely, $ R_{ab}(u) =R(u)\otimes \II$, $R_{bc}(u) =\II \otimes R(u)$, $R_{ac}(u)=P_{bc}R_{ab}(u)P_{bc}$, and~$\II$ is the identity operator in~$\CC^2$.

The $R$-matrix \eqref{Rm} is $\mathfrak{gl}_2$-invariant, and thus,
\begin{gather*}
[R_{ab}(u),K_a K_b]=0,
\end{gather*}
for any matrix $K\in \operatorname{End}\big(\CC^2\big)$.

Starting from this $R$-matrix, we can define a quantum group algebra, so called Yangian of~$\mathfrak{gl}_2$. For this, we introduce a monodromy matrix
\begin{gather*}
T(u)=\left(\begin{matrix} t_{11}(u)& t_{12}(u)\\ t_{21}(u)&t_{22}(u)\end{matrix}\right).
\end{gather*}
This matrix satisfies the RTT relation
\begin{gather}\label{RTT}
 R_{ab}(u-v)T_a(u)T_b(v)=T_b(v)T_a(u) R_{ab}(u-v).
\end{gather}
Equation \eqref{RTT} encodes the commutation relations of the entries $t_{ij}(u)$ (see Appendix~\ref{ComY}), which generate the Yangian algebra.

The transfer matrix
\begin{gather*} 
 t(u)=\Tr_a\big(K_aT_a(u)\big)
\end{gather*}
generates a commutative subalgebra of the Yangian of $\mathfrak{gl}_2$, i.e., $[t(u),t(v)]=0$. Its expansion as a~series in~$u$ leads to the conserved charges of the model that depend on the specific representation. To perform the MABA we need to consider a finite dimensional representation. However, this is not necessarily needed for the usual ABA.

Let us define a highest weight representation by $\cV(\lambda_1(u), \lambda_2(u))$, where $\lambda_i(u)$ are some complex valued functions, and the highest weight vector $\rvec$ is defined by
\begin{gather}\label{HWRG}
t_{ii}(u)\rvec=\lambda_i(u)\rvec, \qquad t_{21}(u)\rvec=0.
\end{gather}
On the other hand, we can define a lowest weight representation $\cV(\lambda_2(u), \lambda_1(u))$, and the lowest weight vector $\rvech$ is defined by
\begin{gather}\label{tbvac}
t_{ii}(u)\rvech=\lambda_{3-i}(u)\rvech, \qquad t_{12}(u)\rvech=0.
\end{gather}

This provides enough tools to consider the usual ABA but not to accomplish the MABA.
In the latter case we need to explicitly fix the representation to be finite dimensional. All of them are given by the highest (lowest) weight representation \cite{Dr}.
More precisely, we consider finite dimensional representations with an arbitrary positive (half)-integer spin $s_i$ for each site of the chain.
In this case the weight vectors and the functions $\lambda_i(u)$ can be constructed from the $\mathfrak{gl}_2$ Lie algebra's spin $s$ finite dimensional representation and the evaluation representation of the Yangian $\mathfrak{gl}_2$.

 The $\mathfrak{gl}_2$ Lie algebra is generated by
 $\{S^z,S^-,S^+,I\}$
that satisfy the commutation relations
\begin{gather}\label{comgl2}
[S^+,S^-]=2S^z,\qquad [S^z,S^\pm]=\pm S^\pm, \qquad [S^i,I]=0.
\end{gather}
The spin $s$ finite dimensional representation of the $\mathfrak{gl}_2$ Lie algebra is given by the highest (lowest) weight representation. Let us define the
highest weight vector of spin $s$ representation by $|s\rangle$. The action of the $\mathfrak{gl}_2$ generators
are given by
\begin{gather*}
S^z|s\rangle=s|s\rangle, \qquad S^+|s\rangle=0,\qquad (S^- )^{2s+1}|s\rangle=0,\qquad I|s\rangle=|s\rangle.
\end{gather*}
The action of the $\mathfrak{gl}_2$ generators on the lowest weight vector $|{-}s\rangle$ of spin $s$ representation is
\begin{gather*}
S^z|{-}s\rangle=-s|{-}s\rangle, \qquad S^-|{-}s\rangle=0,\qquad (S^+ )^{2s+1}|{-}s\rangle=0,\qquad I|{-}s\rangle=|{-}s\rangle.
\end{gather*}
The lowest and the highest vector are related by the formula
\begin{gather*}
 |{-}s\rangle= \frac{ (S^- )^{2s}}{\prod\limits_{j=1}^{2s}\sqrt{j(2s-(j-1))}}|s\rangle.
\end{gather*}
To construct the representation of the Yangian of $\mathfrak{gl}_2$ we take the highest (lowest) weight vector as a tensor product of~$\mathfrak{gl}_2$ highest (lowest) weight vectors with different spins. They are given by
\begin{gather*}
\rvec=\bigotimes_{i=1}^N |s_i\rangle,\qquad \rvech=\bigotimes_{i=1}^N |{-}s_i\rangle.
\end{gather*}
The associated weight functions are
\begin{gather*}
\lambda_1(u)=\prod_{i=1}^N\frac{u-\theta_i+c\big(s_i+\frac{1}{2}\big)}{c}, \qquad \lambda_2(u)=\prod_{i=1}^N\frac{u-\theta_i-c\big(s_i-\frac{1}{2}\big)}{c}.
\end{gather*}
The evaluation representation is given by the Lax operator
\begin{gather*}
\cL_{ai}(u)=\left(\begin{matrix} \dfrac{u}{c}+\dfrac{1}{2}+S_i^3& S_i^-\vspace{1mm}\\ S_i^+&\dfrac{u}{c}+\dfrac{1}{2}-S_i^3\end{matrix}\right)_a,
\end{gather*}
in each site of the spin chain (see, e.g.,~\cite{Skl85}). Due to the commutation relations~\eqref{comgl2}, the Lax operator satisfies the RTT relation~\eqref{RTT} and the $\mathfrak{gl}_2$ invariance~\cite{Skl85}
\begin{gather*}
[\cL_{ai}(u),K_a\cK_i]=0.
\end{gather*}
Here
\begin{gather*}
\cK_i=k_0 \exp{\big(k^+S_i^{+}+k^-S_i^{-}+2 k_3S_i^{3}\big)},\\
K=k_0 \exp{\big(k^+\sigma^{+}+k^-\sigma^{-}+k_3\sigma^{3}\big)}=\left(\begin{matrix} \tilde \kappa & \kappa^+\\ \kappa^-&\kappa \end{matrix}\right),
\end{gather*}
and $\sigma^\pm$, $\sigma^z$ are Pauli matrices.

The monodromy matrix for $N$ arbitrary positive (half)-integer spins $\bar s=\{s_1,\dots,s_N\}$ is then given by
\begin{gather}\label{Tmonorep}
T_a(u)=\cL_{a1}(u-\theta_1)\cdots \cL_{aN}(u-\theta_N),
\end{gather}
where $\{\theta_1,\dots,\theta_N\}$ are inhomogeneity parameters. The monodromy matrix contains the $\mathfrak{gl}_2$ Lie algebra generators $\{S^+,S^-,S^z,I\}$ realized as $S^\alpha=\sum\limits_{i=1}^NS^\alpha_i$,
\begin{gather*}
T(u)= \left(\frac uc\right)^N+\left(\left(\begin{matrix} \dfrac{I+S^z}{2}& S^-\\ S^+&\dfrac{I-S^z}{2}\end{matrix}\right)
-\sum_{i=1}^N\theta_i \right)\left(\frac uc\right)^{N-1}+\cdots.
\end{gather*}
Using the RTT relation we can extract the relation between the $\mathfrak{gl}_2$ Lie algebra generators and the $t_{ij}(u)$, in particular we have:
\begin{gather*}
[S^z,t_{ii}(u)]=0, \qquad [S^z,t_{12}(u)]=t_{12}(u), \qquad [S^z,t_{21}(u)]=-t_{21}(u).
\end{gather*}
We say that the transfer matrix has a ${\rm U}(1)$ symmetry if $[S^z,t(u)]=0$, that is, in the case of a~diagonal twist $K$. In fact, if $\Det(K)\neq 0$, then due to the global $\mathfrak{gl}_2$ invariance we can always restore the ${\rm U}(1)$ symmetry, by diagonalizing the twist matrix and doing a proper similarity transformation~\cite{RMG03}, to perform usual ABA\footnote{In the case of the XXZ models one should consider a more sophisticated face-vertex transformation of the twist (see, e.g.,~\cite{BP152, ZLCYSW15b} and references therein).}. Here, as we want to study the MABA, we do not use this possibility.

Starting form the highest or the lowest weight vector we can construct two possible basis, that we call Bethe vectors, to span the Hilbert space of the model $\bigotimes_{i=1}^N\CC^{2s_i+1}$. Let $\bar u$ denote a~set of arbitrary complex parameters $\{u_1,\dots ,u_M\}$ with cardinality $\#\bar u=M$, and let $S=\sum\limits_{i=1}^N 2 s_i $. The representation space of the Yangian can be spanned by
\begin{gather}\label{diagBV}
\bbb^M(\bar u)=\prod_{i=1}^Mt_{12}(u_i)\rvec
\end{gather}
with $\# \bar u=M$ and $M=0,1,\dots ,S$. Or equivalently, by
\begin{gather}\label{diagBVh}
 \hat\bbb^{\hat M}(\bar v)=\prod_{i=1}^{\hat M}t_{21}(v_i)\rvech
\end{gather}
with $\# \bar v=\hat M$ and $\hat M=0,1,\dots ,S $.

These two basis are related by a morphism of the Yangian. Let us consider the mappings
\begin{gather*}
\psi_1(T(u))= T(-u), \qquad \psi_2(T(u))= T^t(u),
\end{gather*}
where $t$ is the usual transposition in the auxiliary space, $A_{ij}^t =A_{ji}$. They define anti-morphism of the RTT relation (see, e.g.,~\cite{Mol02}). Taking the composition of these mappings we obtain an automorphism of the Yangian
\begin{gather}\label{auto-mor}
\phi(T(u))=\psi_1 \circ \psi_2(T(u))= T^t(-u).
\end{gather}
This mapping relates the two different bases of the representation space of the Yangian. Indeed we have
\begin{gather*}
\phi\big(\bbb^M(\bar u)\big)=\hat \bbb^M(-\bar u),
\end{gather*}
where we have to apply the rule $\phi(\lambda_i(u))=\lambda_{3-i}(-u)$ and $\phi(\rvec)=\rvech$ to preserve the structure of the action. Thus, $\phi$ relates the highest weight and the lowest weight representations. We remark also that $\phi(t(u))=t(-u)$, and therefore the action of the transfer matrix on one of the basis defines the action on the other.	

\section{Algebraic Bethe ansatz and quantum Wronskian equation}\label{S:ABA}

Let us recall the algebraic Bethe ansatz \cite{BPRS12b, FST79,FT84} for generalized $\mathfrak{gl}_2$ invariant integrable models with ${\rm U}(1)$ symmetry. For this we introduce new notation. Let us define rational functions
\begin{gather*}
g(u,v)=\frac{c}{u-v},\qquad f(u,v)=1+g(u,v)=\frac{u-v+c}{u-v},
\end{gather*}
and a shorthand notation for the products of the $f$ functions and the $g$ functions
\begin{gather} \label{SH-not}
g(z,\bar u)=\prod_{i=1}^{M}g(z, u_i), \qquad f(\bar u,z)=\prod_{i=1}^{M}f(u_i,z), \qquad f(u_i,\bar u_i)=\prod_{j=1, j\neq i}^{M}f(u_i, u_j),
\end{gather}
where $\bar u_i=\{\bar u\setminus u_i\}$ is the set complementary to the element $u_i$.

The ABA allows to solve the spectral problem for the transfer matrix
\begin{gather*}
t_d(z)=\Tr_a\big(K_a^{(d)}T_a(z)\big)=\kb t_{11}(z)+ \kappa t_{22}(z),
\end{gather*}
with the diagonal twist matrix
\begin{gather}\label{diagtwist}
 K^{(d)}=\left(\begin{matrix} \kb & 0\\0&\kappa \end{matrix}\right).
\end{gather}
In this case we have the ${\rm U}(1)$ symmetry regarding to the operator $S^z$, i.e., $[S^z,t_d(z)]=0$ and simple eigenstates are given by the highest and lowest weight vectors $\{\rvec, \rvech\}$
\begin{gather}\label{act-td-WV}
t_d(z)\rvec=( \kb\lambda_1(z)+ \kappa \lambda_2(z))\rvec, \qquad t_d(z)\rvech=( \kappa \lambda_1(z)+ \kb \lambda_2(z))\rvech.
\end{gather}
We remark that for $K^{(d)}\propto I $, i.e., $\kappa=\kb$, that corresponds to the periodic boundary condition, both eigenvalues are the same and the model has additional symmetries (full $\mathfrak{gl}_2$ invariance).
We should recall that the method applies for any highest weight representation~\eqref{HWRG} without specifying the functions~$\lambda_i(u)$. In particular it also applies for an infinite dimensional representation as in the case of the Bose gas model.

We have seen in the previous section that the space of states can be spanned by Bethe vectors~\eqref{diagBV} for a given set of generic variables $\bar u=\{u_1,\dots,u_M\}$ of the cardinality $\# \bar u=M$. Due to the ${\rm U}(1)$ symmetry the spectrum is decomposed regarding to subspaces of the total spin~$S^z$. The Bethe vectors~\eqref{diagBV} form a basis for the total spin
 \begin{gather*}
 S^z\bbb^M(\bu)=\left( \frac{S}{2}-M\right)\bbb^M(\bu).
 \end{gather*}
Then $M$ is a given quantum number and we can diagonalize the transfer matrix for each invariant subspace labeled by $M$.

\begin{thm}[\cite{FST79}] The action of the transfer matrix onto the Bethe vector \eqref{diagBV} is given by
\begin{gather*}
t_d(z)\bbb^M(\bu)=\Lambda_d^M(z,\bar u)\bbb^M(\bu)+\sum_{i=1}^Mg(u_i,z)E_d^M(u_i,\bar u_i)t_{12}(z)\bbb^{M-1}(\bu_i),
\end{gather*}
where
\begin{gather*}
\Lambda_d^M(z,\bar u)=\kb \lambda_1(z)f(\bar u,z)+\kappa \lambda_2(z)f(z,\bar u), \\
E_d^M(u_i,\bar u_i)=-\kb \lambda_1(u_i)f(\bar u_i,u_i)+\kappa \lambda_2(u_i)f(u_i,\bar u_i).
\end{gather*}
If the set $\bu$ satisfies the Bethe equations
\begin{gather*}
\frac{\kb \lambda_1(u_i)}{ \kappa \lambda_2(u_i)} =\frac{f(u_i,\bar u_i)}{f(\bar u_i,u_i)}=-\prod_{j=1}^M\frac{u_i-u_j+c}{u_i-u_j-c},
\end{gather*}
for $i=0,1,\dots,M$, then the Bethe vectors \eqref{diagBV} are eigenvectors of the transfer matrix and are called on-shell Bethe vectors.
\end{thm}

\begin{proof} Multiplication of the original monodromy matrix by the diagonal twist reduces to the redefinition of the functions $\lambda_i(u)$. Since within the framework of ABA these functions are not fixed, the proof then follows from the standard considerations of the ABA~\cite{FST79}. It is based on the commutation relations of the operators $t_{ij}(u)$ (see Appendix~\ref{ComY}) and the action on the highest weight vector~\eqref{HWRG}.
\end{proof}

Using the mapping $\phi$ \eqref{auto-mor} we can find the transfer matrix action on the Bethe vectors constructed from the lowest weight vector. This gives another parametrization for the spectrum of the transfer matrix.

\begin{cor} The action of the transfer matrix onto the Bethe vector $\hat \bbb^{\hat M}(\bv)$ is given by
\begin{gather*}
t_d(z)\hat \bbb^{\hat M}(\bv)=\hat\Lambda_d^M(z,\bar v)\hat \bbb^{\hat M}(\bv)+\sum_{i=1}^{\hat M}g(z,v_i)\hat E_d^{\hat M}(v_i,\bar v_i)t_{21}(z)\hat \bbb^{\hat M-1}(\bv_i),
\end{gather*}
where
\begin{gather*}
\hat\Lambda_d^{\hat M}(z,\bar v)=\kappa \lambda_1(z)f(\bar v,z)+\kb \lambda_2(z)f(z,\bar v), \\
\hat E_d^{\hat M}(v_i,\bar v_i)=-\kappa \lambda_1(v_i)f(\bar v_i,v_i)+\kb \lambda_2(v_i)f(v_i,\bar v_i).
\end{gather*}
If the set $\bv$ satisfies the second family of Bethe equations
\begin{gather*}
\frac{\kappa \lambda_1(v_i)}{ \kb \lambda_2(v_i)} =\frac{f(v_i,\bar v_i)}{f(\bar v_i,v_i)}=-\prod_{j=1}^{\hat M}\frac{v_i-v_j+c}{v_i-v_j-c}
\end{gather*}
 for $i=0,1,\dots , \hat M$, then the on-shell Bethe vectors \eqref{diagBVh} are eigenvectors of the transfer matrix.
\end{cor}

We can rewrite the two results of the previous section in terms of Baxter Q-polynomials. Let us define the polynomials
\begin{gather*}
q_+^M(z)=\prod_{i=1}^M\frac{z-u_i}{c}=\big(g(z,\bar u)\big)^{-1}
\end{gather*}
and
\begin{gather*}
q_-^{\hat M}(z)=\prod_{i=1}^{\hat M}\frac{z-v_i}{c}=\big(g(z,\bar v)\big)^{-1}.
\end{gather*}
This allows us to rewrite the spectral problem in terms of Baxter T-Q equations
\begin{gather}\label{tqp}
\Lambda_d(z) q_+^{M}(z)=\kb \lambda_1(z)q_+^{M}(z-c)+\kappa \lambda_2(z)q_+^{M}(z+c),
\end{gather}
and
\begin{gather}\label{tqm}
\Lambda_d(z) q_-^{\hat M}(z)=\kappa \lambda_1(z)q_-^{\hat M}(z-c)+\kb \lambda_2(z)q_-^{\hat M}(z+c).
\end{gather}
The two characterizations of the spectrum by Baxter T-Q equation satisfy a quantum Wronskian condition.
\begin{thm} The two Baxter Q-polynomials $q_+^M(z)$ and $q_-^{\hat M}(z)$ are related by the quantum Wronskian equation~{\rm \cite{PS99}}
\begin{gather}\label{WEd}
W_d(z)-W_d(z+c)=0,
\end{gather}
with
\begin{gather*}
W_d(z)=\frac{\kb q_+^{\hat M}(z-c)q_-^{M}(z)-\kappa q_+^{\hat M}(z)q_-^{M}(z-c)}{\lambda(z)},
\end{gather*}
and
\begin{gather*}
\frac{\lambda_{1}(z)}{\lambda_{2}(z)}=\frac{\lambda(z+c)}{\lambda(z)}, \qquad
\lambda(z)=\prod_{i=1}^N\prod_{k=1}^{2s_i}\frac{z-\theta_i+c\big(s_i-k+\frac{1}{2}\big)}{c}.
\end{gather*}
This implies $\hat M+M=S=\sum\limits_{i=1}^N2s_i$ and $W_d(z)= (\kb-\kappa)$.
\end{thm}

\begin{proof}
For a given eigenvalue $\Lambda_d(z)$ of the transfer matrix, the quantum Wronskian equation follows from the difference between \eqref{tqp} times $q_-^{\hat M}(z)$ and \eqref{tqm} times $q_+^{M}(z)$, and from the fact that $\frac{\lambda_{1}(z)}{\lambda_{2}(z)}=\frac{\lambda(z+c)}{\lambda(z)}$. Then, due to \eqref{WEd} we conclude that $W_d(z)$ is a constant. Taking the limit $z\to \infty$ we find that $\hat M+M=S$ and $W_d(z)= (\kb-\kappa)$.
\end{proof}

\section[Modified algebraic Bethe ansatz and quantum Wronskian equation]{Modified algebraic Bethe ansatz\\ and quantum Wronskian equation}\label{S:MABA}

For a non-diagonal twist, the transfer matrix
\begin{gather}\label{ndtrans}
t(z)=\Tr_a\big(K_aT_a(z)\big)=\kb t_{11}(z)+\ka t_{22}(z)+\kp t_{21}(z)+\km t_{12}(z)
\end{gather}
with
\begin{gather*}
K=\left(\begin{matrix} \tilde \kappa & \kappa^+\\ \kappa^-&\kappa \end{matrix}\right),
\end{gather*}
does not commute with the operator $S^z$. We have $[S^z,t(z)]= \kp t_{21}(z)-\km t_{12}(z)$. As a~consequence, the highest and lowest weight vectors are not anymore eigenvectors of the transfer matrix. The action of the twisted transfer matrix on the highest weight vector~\eqref{HWRG}, which is not anymore eigenstate unless $\km=0$, is given by
\begin{gather*}
t(z)\rvec=(\kb \lambda_1(z)+\ka \lambda_2(z))\rvec+\km t_{12}(z)\rvec.
\end{gather*}
Similarly the lowest weight vector~\eqref{tbvac} is not an eigenstate either (unless $\kp=0$), and the action of the twisted transfer matrix on it is
\begin{gather*}
t(z)\rvech=(\kb \lambda_2(z)+\ka \lambda_1(z))\rvech+ \kp t_{21}(z)\rvech.
\end{gather*}

The MABA allows us to construct the modified Bethe vector keeping the highest or lowest weight vectors as a starting point. The idea of the modified algebraic Bethe ansatz \cite{ABGP15,Bel14,BP152,BC13,BP15} relies in the construction of modified operators that preserve the operator algebra structure\footnote{For XXZ spin chain the underling operator algebra structure is not preserved but mapped in a more complex one (see \cite{ABGP15, Bel14,BP152} and references therein).}. The transformation of the monodromy matrix
\begin{gather*}
\bar T(z)=AT(z)B=\left(\begin{matrix} \nu_{11}(z)& \nu_{12}(z)\\ \nu_{21}(z)&\nu_{22}(z)\end{matrix}\right),
\end{gather*}
where $A$ and $B$ are two arbitrary two by two matrices, is an automorphism of the Yangian of $\mathfrak{gl}_2$, i.e., new operators satisfy the same commutation relations as the operators~$t_{ij}(z)$ (see Appendix~\ref{ComY}). We remark that the modified operators can be seen as a transfer matrix for a~model with a twist that has null determinant:
\begin{gather*}
\nu_{ij}(z)=\Tr_a\big((V_{ji})_aT_a(z))
\end{gather*}
with
\begin{gather*}
V_{ji}=BE_{ji}A,
\end{gather*}
that satisfies $\Det(V_{ji})=0$. We call such operators {\it null twisted transfer matrices}. The Yangian generators~$t_{ij}(u)$, the modified operators~\cite{BC13,BP15}, as well as `good' operators~\cite{Bgood} are such objects.

We look for a deformation of the action of the transfer matrix on the highest weight vector in terms of the modified creation operator $\nu_{12}(u)$ of the following form:
\begin{gather}\label{newacth}
t(z)\rvec=\big((\kb-\rho_1) \lambda_1(z)+(\ka-\rho_2) \lambda_2(z)\big)\rvec+ \eta \nu_{12}(z)\rvec.
\end{gather}
Here $\rho_i$ and $\eta$ are some scalars that can be seen as deformation parameters that go to zero when we restore the ${\rm U}(1)$ symmetry and reduce~\eqref{newacth} to~\eqref{act-td-WV}. The matrix~$V_{21}$ is uniquely determined in terms of $\{\ka,\kb,\km,\kp,\rho_1,\rho_2,\eta\}$.

We do the same for the action on the lowest weight vector and look for the modified creation operator $\nu_{21}(z)$ such that
\begin{gather*}
t(z)\rvech=\big((\kb-\rho_1) \lambda_2(z)+(\ka-\rho_2) \lambda_1(z)\big)\rvech+ \hat \eta \nu_{21}(z)\rvech.
\end{gather*}
The matrix $V_{12}$ is uniquely determined in terms of $\{\ka,\kb,\km,\kp,\rho_1,\rho_2,\hat \eta\}$. We fix also the factorisation of the twist
\begin{gather}\label{Kdecom}
K=BDA,
\end{gather}
with
\begin{gather*}
D=\left(\begin{matrix} \kb-\rho_1&0\\ 0&\ka-\rho_2\end{matrix}\right).
\end{gather*}
Thus, we have the diagonal modified transfer matrix
\begin{gather*}
t(z)=\Tr(D\bar T(z))=(\kb-\rho_1)\nu_{11}(z)+(\ka-\rho_2)\nu_{22}(z).
\end{gather*}
Then, it is easy to see that $A$ and $B$ are given by\footnote{We have some freedom that follows from the transformation $A\to C^{-1}A$ and $B\to BC$ with $[C,D]=0$ that leaves \eqref{Kdecom} invariant.}
\begin{gather*}
A=\sqrt\mu\left(\begin{matrix} 1&\dfrac{\rho_2}{\km}\\ \dfrac{\rho_1}{\kp}&1\end{matrix}\right),\qquad B=\sqrt\mu\left(\begin{matrix} 1&\dfrac{\rho_1}{\km}\\ \dfrac{\rho_2}{\kp}&1\end{matrix}\right),\qquad \mu=\frac{1}{1-\frac{\rho_1\rho_2}{\kp\km}},
\end{gather*}
provided $\rho_1$ and $\rho_2$ enjoy the following relation:
\begin{gather*}
 \kp\km-(\rho_1 \kappa+\rho_2 \kb)+\rho_1\rho_2=0.
\end{gather*}

We can now express the modified operators as linear combination of the original opera\-tors~$t_{ij}(z)$
\begin{gather}\label{nutot11}
\nu_{11}(z)=\mu \left(t_{11}(z)+\frac{\rho_2}{\kp}t_{12}(z)+\frac{\rho_2}{\km}t_{21}(z)+\frac{\rho_2^2}{\km\kp}t_{22}(z)\right),\\
\label{nutot22}
\nu_{22}(z)=\mu \left(t_{22}(z)+\frac{\rho_1}{\kp}t_{12}(z)+\frac{\rho_1}{\km}t_{21}(z)+\frac{\rho_1^2}{\km\kp}t_{11}(z)\right),\\
\label{nutot12}
\nu_{12}(z)=\mu \left(t_{12}(z)+\frac{\rho_1}{\km}t_{11}(z)+\frac{\rho_2}{\km}t_{22}(z)+\frac{\rho_1\rho_2}{(\km)^2}t_{21}(z)\right),\\
\label{nutot21}
\nu_{21}(z)=\mu \left(t_{21}(z)+\frac{\rho_1}{\kp}t_{11}(z)+\frac{\rho_2}{\kp}t_{22}(z)+\frac{\rho_1\rho_2}{(\kp)^2}t_{12}(z)\right),
\end{gather}
with
\begin{gather*}
\mu=\frac{1}{1-\frac{\rho_1\rho_2}{\km\kp}}.
\end{gather*}

\begin{rmk}
In the case $\rho_1=\rho_2=\rho$ we recover the decomposition of the twist matrix given in \cite{BP15} with $A=B=L$.
\end{rmk}

\begin{rmk} With the constraint $\kp=\km=0$ we recover the diagonal twist \eqref{diagtwist} and the action formulas~\eqref{Off-t-M-vac}--\eqref{Off-t-E1} take the usual ABA form. From the point of view of the decomposition of the twist matrix $K= BDA$ we can recover the diagonal twist in two ways. On the one hand, taking $\rho_i=0$ and then $\kp=\km=0$ we send the modified operators $\nu_{ij}(z)$ to the initial operators $t_{ij}(z)$. On the other hand, taking first $ \kp=\km$, redefining $\rho_i=\kp \bar \rho_i$ with $\bar \rho_i\neq 0$ and then specifying $\kp=0$ we keep the modified operators $\nu_{ij}(z)$ provided $ \bar \rho_1 \kappa+\bar \rho_2 \kb=0$. The~$B^{\rm good}$ operator in~$\mathfrak{gl}_2$ case proposed in~\cite{Bgood} belongs to this special case of the modified operators.
\end{rmk}

Using \eqref{nutot11}, \eqref{nutot22} we can calculate the actions of the modified operators on the highest weight and the lowest weight vectors.

\begin{prop}If we consider $\nu_{12}(z)$ as a modified creation operator, then the actions of the modified operators $\{\nu_{11}(z),\nu_{22}(z),\nu_{21}(z)\}$ on the highest weight vector~\eqref{HWRG} are given by
\begin{gather}
\nu_{11}(z)\rvec= \lambda_1(z)\rvec+\frac{\rho_2}{\kp}\nu_{12}(z)\rvec,\nonumber\\
\nu_{22}(z)\rvec= \lambda_2(z)\rvec+\frac{\rho_1}{\kp}\nu_{12}(z)\rvec,\label{nvac12}\\
\nu_{21}(z)\rvec=\left(\frac{\rho_1}{\kp}\lambda_1(z)+\frac{\rho_2}{\kp}\lambda_2(z)\right)\rvec
+\frac{\rho_1\rho_2}{(\kp)^2}\nu_{12}(z)\rvec.\nonumber
\end{gather}
On the other hand, if we choose $\nu_{21}(z)$ as a modified creation operator, then the actions of the modified operators $\{\nu_{11}(z),\nu_{22}(z),\nu_{12}(z)\}$ on the lowest weight vector \eqref{tbvac} are given by
\begin{gather*}
\nu_{11}(z)\rvech= \lambda_2(z)\rvech+\frac{\rho_2}{\km}\nu_{21}(z)\rvech,\\
\nu_{22}(z)\rvech= \lambda_1(z)\rvech+\frac{\rho_1}{\km}\nu_{21}(z)\rvech,\\
\nu_{12}(z)\rvech=\left( \frac{\rho_2}{\km}\lambda_1(z)+\frac{\rho_1}{\km}\lambda_2(z)\right)\rvech+\frac{\rho_1\rho_2}{(\km)^2}\nu_{21}(z)\rvech.
\end{gather*}
\end{prop}

\begin{proof} The proposition follows from the explicit form of the modified operators in terms of the original operators $t_{kl}(z)$ given by \eqref{nutot11}--\eqref{nutot21} and the action on the weight vec\-tors~\eqref{HWRG},~\eqref{tbvac}.
\end{proof}

\begin{rmk}
For a given weight vector $\rvec$ or $\rvech$ the role of $\nu_{12}(z)$ and $\nu_{21}(z)$ can be inverted. Indeed, as example, from the last equation of \eqref{nvac12} we can express $\nu_{12}(z)\rvec$ in term of $\nu_{21}(z)\rvec$ and $\rvec$ and then consider $\nu_{21}(z)$ as a creation operator on $\rvec$.
\end{rmk}

It follows from the action of the modified operators that the actions of the transfer matrix on the weight vectors are given by
\begin{gather*}
t(z)\rvec=\big(( \kb -\rho_1)\lambda_1(z)+( \kappa -\rho_2)\lambda_2(z)\big)\rvec+\frac{\km}{\mu}\nu_{12}(z)\rvec,
\end{gather*}
and
\begin{gather*}
t(z)\rvech=\big(( \kb -\rho_1)\lambda_2(z)+( \kappa -\rho_2)\lambda_1(z)\big)\rvech+\frac{\kp}{\mu}\nu_{21}(z)\rvech.
\end{gather*}

Let us extend the shorthand notation \eqref{SH-not} to the product of commuting operators, for example,
\begin{gather*}
\nu_{12}(\bar u)=\prod_{i=1}^M\nu_{12}(u_i), \qquad \nu_{21}(\bar v)=\prod_{i=1}^{\hat M}\nu_{21}(v_i).
\end{gather*}
Then the action of $t(z)$ on the modified Bethe vector $\nu_{12}(\bar u)\rvec$ with $\# \bar u=M$ has been derived in~\cite{BP15} and is given by
\begin{gather}
t(z)\nu_{12}(\bar u)\rvec = \frac{\km}{\mu}\nu_{12}(z)\nu_{12}(\bar u)\rvec+\Lambda_1^M(z,\bar u)\nu_{12}(\bar u)\rvec \nonumber\\
\hphantom{t(z)\nu_{12}(\bar u)\rvec =}{} +\sum_{i=1}^{M}g(u_i,z)E_1^M(u_i,\bar u_i)\nu_{12}(z)\nu_{12}(\bar u_i)\rvec,\label{Off-t-M-vac}
\end{gather}
where
\begin{gather}
\Lambda_1^M(z,\bar u)=( \kb -\rho_1)\lambda_1(z)f(\bar u,z)+( \kappa -\rho_2)\lambda_2(z)f(z,\bar u),\\
E_1^M(u_i,\bar u_i)=( \kappa -\rho_2)\lambda_2(u_i)f(u_i,\bar u_i)-( \kb -\rho_1)\lambda_1(u_i)f(\bar u_i,u_i).
\end{gather}
On the other hand, the action of $t(z)$ on the second family of Bethe vectors given by $\nu_{21}(\bar v)\rvech$
with $\# \bar v=\hat M$ reads
\begin{gather}
t(z)\nu_{21}(\bar v)\rvech=\frac{\kp}{\mu}\nu_{21}(z)\nu_{21}(\bar v)\rvech+\hat \Lambda_1^{\hat M}(z,\bar v)\nu_{12}(\bar v)\rvech\nonumber\\
\hphantom{t(z)\nu_{21}(\bar v)\rvech=}{} +\sum_{i=1}^{\hat M}g(v_i,z)\hat E_1^{\hat M}(v_i,\bar v_i)\nu_{21}(z)\nu_{21}(\bar v_i)\rvech,\label{Off-t-M-vach}
\end{gather}
where
\begin{gather}
\hat \Lambda_1^{\hat M}(z,\bar v)=( \kappa -\rho_2)\lambda_1(z)f(\bar v,z)+( \kb -\rho_1)\lambda_2(z)f(z,\bar v),\\
\hat E_1^{\hat M}(v_i,\bar v_i)=( \kb -\rho_1)\lambda_2(v_i)f(v_i,\bar v_i)-( \kappa -\rho_2)\lambda_1(v_i)f(\bar v_i,v_i). \label{Off-t-E1}
\end{gather}

For the finite dimensional representation given by the monodromy matrix~\eqref{Tmonorep}, the multiple product of a~null transfer matrix $\nu(z)=\Tr(V T(z))$ with $\Det(V)=0$ (such as the modified creation and annihilation operators) satisfies the following theorem.

\begin{thm} \label{Th1prodS} Let $\bu$ be a set of arbitrary parameters of cardinality $\# \bu=\sum\limits_{i=1}^N 2 s_i=S$. For a~finite dimensional representation of the monodromy matrix given by~\eqref{Tmonorep}, the following operator identity holds
\begin{gather} \label{Th1prodSeq}
\nu(z)\nu(\bu)=\Tr(V)\left(F(z)g(z,\bu)\nu(\bu)+\sum_{i=1}^S g(u_i,z)F(u_i)g(u_i,\bu_i)\nu(z)\nu(\bu_i)\right),
\end{gather}
where $\nu(\bu_i)= \prod\limits^{S}_{j=1, \,j\neq i}\nu(u_j)$ and
\begin{gather}\label{F-rep}
F(z)=\prod_{i=1}^N\prod_{k=0}^{2s_i} \frac{z-\theta_i+c \big(s_i-k+\frac{1}{2}\big)}{c}.
\end{gather}
The l.h.s.\ of \eqref{Th1prodSeq} involves the product of $S+1$ operators $\nu$ and each term of the r.h.s.\ involves the product of $S$ operators~$\nu$.
\end{thm}

\begin{proof} It is proved in the Appendix~\ref{simplenu} that for $\Tr(V)\neq 0$ and for arbitrary inhomogeneity parameters $\bar\theta=\{\theta_1,\dots,\theta_N\}$, the operator $\nu(z)$ has simple spectrum. Moreover, its inverse multiplied by the function $F(z)$ has polynomial eigenvalues of degree $S=\sum\limits_{i=1}^N2 s_i$ with the leading term given by
\begin{gather*}
F(z)\nu^{-1}_{12}(z)=\Tr(V)^{-1}\left(\frac zc\right)^{S}+\cdots , \qquad z\to \infty.
\end{gather*}
Let $\#\bar u=S+1$. Consider a product of operators
\begin{gather*}
F(z)\nu^{-1}_{12}(z)\nu(\bar u)g(z,\bar u).
\end{gather*}
It has simples poles at the points $u_i$ and behaves as $z^{-1}$ at infinity. Then, taking the sum of all residues we find that
\begin{gather*}
\nu(\bar u)=\Tr(V)\sum_{i=1}^{S+1}F(u_i)g(u_i,\bar u_i)\nu(\bar u_i).
\end{gather*}
Setting $u_{S+1}=z$ we complete the proof.
\end{proof}

\begin{cor} \label{mulmodop} We can specify Theorem~{\rm \ref{Th1prodS}} by taking $V=BE_{21}A$ to find
\begin{gather*} 
\frac{\kappa^-}{\mu}\nu_{12}(z)\nu_{12}(\bu)= (\rho_1+\rho_2)\left(\!F(z)g(z,\bu)\nu_{12}(\bu)+\sum_{i=1}^S g(u_i,z)F(u_i)g(u_i,\bu_i)\nu_{12}(z)\nu_{12}(\bu_i)\!\right),
\end{gather*}
and by taking $V=BE_{12}A$ to find
\begin{gather*} 
\frac{\kappa^+}{\mu}\nu_{21}(z)\nu_{21}(\bv)= (\rho_1+\rho_2)\left(F(z)g(z,\bv)\nu_{21}(\bv)+\sum_{i=1}^S g(v_i,z)F(v_i)g(v_i,\bv_i)\nu_{21}(z)\nu_{21}(\bv_i)\right).
\end{gather*}
\end{cor}

These examples show the algebraic origin of the inhomogeneous term of the modified Baxter T-Q equation \cite{CYSW13a,CYSW13b,CYSW13c}. It was conjectured and then proved in \cite{ABGP15,Bel14,BP152,BP15,Cram14} for models on a~segment. Here we go beyond and prove this property independently of the action on the highest weight vector in the case the twisted XXX spin chain with arbitrary positive (half)-integer spins.

\begin{rmk}
For the fundamental representation, $s_i=\frac{1}{2}$, we have $F(z)=\lam_{1}(z)\lam_{2}(z)$. This proves the conjecture of \cite{BP15}.
\end{rmk}

\begin{rmk}The entries of the monodromy matrix can be treated as null transfer matrices $t_{ij}(z)=\Tr(E_{ji}T(z))$. Then we have from Theorem~\ref{Th1prodS}
\begin{gather*}
t_{ii}(z)t_{ii}(\bu)=F(z)g(z,\bu)t_{ii}(z)t_{ii}(\bu)+\sum_{i=1}^S F(u_i)g(u_i,z)g(u_i,\bu_i)t_{ii}(z)t_{ii}(\bu_i), \\
t_{ij}(z)t_{ij}(\bu)=0,
\end{gather*}
for $i\neq j$, $\# \bu=S$, and the function $F(z)$ given by~\eqref{F-rep}.
\end{rmk}

\begin{thm}For finite dimensional representation of the monodromy matrix given by~\eqref{Tmonorep}, the action of the transfer matrix on the Bethe vector $\nu_{12}(\bar u)\rvec$ with $ \#\bu=\sum\limits_{i=1}^N 2s_i=S$ is given by
\begin{gather*} 
t(z)\nu_{12}(\bar u)\rvec=\Lambda(z,\bu)\nu_{12}(\bar u)\rvec +\sum_{i=1}^{S} g(u_i,z)E(u_i,\bu_i)\nu_{12}(z)\nu_{12}(\bar u_i)\rvec,
\end{gather*}
where
\begin{gather*}
\Lambda(z,\bu)=( \kb -\rho_1)\lambda_1(z)f(\bar u,z)+( \kappa -\rho_2)\lambda_2(z)f(z,\bar u)+(\rho_1+\rho_2)F(z)g(z,\bu),\\
E(u_i,\bu_i)=( \kappa -\rho_2)\lambda_2(u_i)f(u_i,\bar u_i)-( \kb -\rho_1)\lambda_1(u_i)f(\bar u_i,u_i)+(\rho_1+\rho_2)F(u_i)g(u_i,\bu_i).
\end{gather*}
The action of the transfer matrix on the Bethe vector $\nu_{21}(\bar v)\rvech$ with $\bv=\sum\limits_{i=1}^N 2s_i=S$ is given by
\begin{gather*}
t(z)\nu_{21}(\bar v)\rvech=\hat\Lambda(z,\bv)\nu_{21}(\bar v)\rvech +\sum_{i=1}^{S} g(v_i,z)\hat E(v_i,\bv_i)\nu_{21}(z)\nu_{21}(\bar v_i)\rvech,
 \end{gather*}
where
\begin{gather*}
\hat\Lambda(z,\bv)=( \ka -\rho_2)\lambda_1(z)f(\bar v,z)+( \kb -\rho_1)\lambda_2(z)f(z,\bar v)+(\rho_1+\rho_2)F(z)g(z,\bv),\\
\hat E(v_i,\bv_i)=( \kb -\rho_1)\lambda_2(v_i)f(v_i,\bar v_i)-( \ka-\rho_2)\lambda_1(v_i)f(\bar v_i,v_i)+(\rho_1+\rho_2)F(v_i)g(v_i,\bv_i).
\end{gather*}
Thus, when the inhomogeneous Bethe equations are satisfied, i.e., $E(u_i,\bu_i)=0$ and $\hat E(v_i,\bv_i)=0$ for $i=1,\dots,S$, the vectors $\nu_{12}(\bar u)\rvec$ and $\nu_{21}(\bar v)\rvech$ are eigenvectors of the transfer matrix.
\end{thm}
\begin{proof} The theorem follows from the actions \eqref{Off-t-M-vac}--\eqref{Off-t-M-vach} and Corollary~\ref{mulmodop}.
\end{proof}

Then we can rewrite the two eigenvalues by two inhomogeneous Baxter T-Q functional equation. We define the functional Q-operators $Q^+(z)=(g(z,\bu))^{-1}$ and $Q^-(z)=(g(z,\bv))^{-1}$, that are two polynomials in $z$ of degree~$S$, to find that
\begin{gather}\label{inTQp}
\Lambda(z)Q^+(z)=( \kb -\rho_1)\lambda_1(z)Q^+(z-c)+( \ka -\rho_2)\lambda_2(z)Q^+(z+c)+(\rho_1+\rho_2)F(z)
\end{gather}
and
\begin{gather}\label{inTQm}
\Lambda(z)Q^-(z)=( \ka -\rho_2)\lambda_1(z)Q^-(z-c)+( \kb -\rho_1)\lambda_2(z)Q^-(z+c)+(\rho_1+\rho_2)F(z).
\end{gather}
Then we can construct quantum Wronskian equations~\cite{PS99}.
\begin{thm} The two functional Baxter Q-operators $Q^\pm(z)$ are related by the modified quantum Wronskian equation given as
\begin{gather*}
W(z)-W(z+c)= (\rho_1+\rho_2)(Q^+(z)-Q^-(z)),\\
W(z)=\frac{(\kb-\rho_1)Q^+(z-c)Q^-(z)-(\kappa-\rho_2)Q^+(z)Q^-(z-c)}{\lambda(z)}.
\end{gather*}
\end{thm}
\begin{proof} For a given eigenvalue $\Lambda(z)$ of the transfer matrix, one should multiply~\eqref{inTQp} and~\eqref{inTQm} respectively by $Q^-(z)$ and $Q^+(z)$, and consider the difference of the resulting expressions. Then the use of $\frac{\lambda_1(z)}{\lambda_2(z)}=\frac{\lambda(z+c)}{\lambda(z)}$ and $\lambda_1(z)\lambda(z)=F(z)$ directly leads to the modified quantum Wronskian equation.
\end{proof}

\begin{rmk} For invertible twist, the transfer matrix~\eqref{ndtrans} can also be characterized from the usual Baxter T-Q equations \cite{RMG03}. Defining $\alpha^\pm$ to be eigenvalues of the twist matrix $K$ that satisfy $\alpha^++\alpha^-=\kb +\kappa$ and $\alpha^+\alpha^-=\Det K$. Then we have
\begin{gather*}
\Lambda(z)q^M(z)=\alpha^+\lambda_1(z)q^M(z-c)+\alpha^-\lambda_2(z)q^M(z+c).
\end{gather*}
We can construct Wronskian type equation between this parametrization and the inhomogeneous one~\eqref{inTQp}
\begin{gather*}
W(z)-W(z+c)=\left( \frac{\alpha^-}{\kb-\rho_1} \right)^{z/c} (\rho_1+\rho_2)q(z),
\end{gather*}
with
\begin{gather*}
W(z)=\left( \frac{\alpha^-}{\kb-\rho_1} \right)^{z/c}\frac{\alpha^+ q(z-c)Q^+(z)-(\kb-\rho_1)q(z)Q^+(z-c)}{\lambda(z)}.
\end{gather*}
Here we used the identity $\alpha^+\alpha^-=(\kb-\rho_1)(\kappa-\rho_2)$.
\end{rmk}

\section{Conclusion}\label{Conc}

In this paper we presented several proofs of the new steps needed to perform the algebraic Bethe ansatz for the models without ${\rm U}(1)$ symmetry. We showed that the appearance of the new term in the Baxter T-Q equation follows from the analysis of the product of the modified creation operators and that it is a general property of the null twisted transfer matrix. It is not necessary to consider the action on a weight vector.

Then the Bethe ansatz characterization of the spectral problem of the XXX Heisenberg spin chain with an arbitrary twist and arbitrary positive (half)-integer spin at each site of the chain is fully understood by means of the MABA. We also derived, for arbitrary positive (half)-integer spins, a modified quantum Wronskian equation that relates two different characterizations of the spectral problem. It should be of interest to relate the modified quantum Wronskian equation with Hirota equation~\cite{FN16}. Moreover finding the numerical solutions of the Bethe equations remains a challenging open problem (see recent development in \cite{HNS13,JZ17,MV16}), and modified quantum Wronskian equation can be used to address it.

The new actions of the modified operators on the weight vectors deserve to be studied in details. They also appear in the context of the separation of variable by introduction of the~$B^{\rm good}$ operators~\cite{Bgood} (see also recent work of two of the authors~\cite{BS18} that considers the~$\mathfrak{gl}_2$ case from the point of view of the ABA). These actions provide formulas for the development of the modified Bethe vector in terms of the original~$t_{12}(u)$ creation operator. The multiple action of the modified operators should lead to generalization of the results of~\cite{BPRS12b}. These multiple action formulas are very useful for the calculation of Bethe vector's scalar products, form factors, and correlation functions and will be given in a forthcoming publication. In the framework of the MABA for the twisted XXX models, these multiple actions lead to new compact formulas for the scalar products of the modified Bethe vectors and should open a path to prove the modified determinant formulas conjectured in~\cite{BP15}.

An important step should be to achieve the MABA for the twisted XXX $\mathfrak{sl}_3$ spin chain and to other types of boundary conditions or other higher rank algebras. For a lower block twist, there is no new difficulties, as the first step of the nested Bethe ansatz preserves the ${\rm U}(1)$ symmetry. However, for the upper block twist some new tools have to be found.

\appendix

\section[Commutation relations of the $t_{ij}(u)$ and multiple actions on the Bethe vector]{Commutation relations of the $\boldsymbol{t_{ij}(u)}$ and multiple actions\\ on the Bethe vector}\label{ComY}

The RTT relation \eqref{RTT} yields the following commutation relations:
\begin{gather*}
t_{ij}(v)t_{ij}(u)=t_{ij}(u)t_{ij}(v),\\
t_{ij}(v)t_{ik}(u)=f(u,v)t_{ik}(u)t_{ij}(v)+g(v,u)t_{ik}(v)t_{ij}(u),\\
t_{ij}(v)t_{kj}(u)=f(v,u)t_{kj}(u)t_{ij}(v)+g(u,v)t_{kj}(v)t_{ij}(u).
\end{gather*}
They imply the following actions on the products of $M$ operators:
\begin{gather}
t_{ij}(v)t_{ik}(\bar u)=f(\bar u,v)t_{ik}(\bar u)t_{ij}(v)+\sum_{i=1}^M g(v,u_i)f(\bar u_i,u_i)t_{ik}(v)t_{ik}(\bar u_i)t_{ij}(u_i),\nonumber\\
\label{comsl22a}
t_{ij}(v)t_{kj}(\bar u)=f(v,\bar u)t_{kj}(\bar u)t_{ij}(v)+\sum_{i=1}^M g({ u_i,v})f(u_i,\bar u_i)t_{{ kj}}(v)t_{{ kj}}(\bar u_i)t_{{ ij}}(u_i).
\end{gather}
The same commutation relations are valid for the modified operators $\nu_{ij}(u)$.

\section[SoV basis for the modified creation operator $\nu_{12}(z)$]{SoV basis for the modified creation operator $\boldsymbol{\nu_{12}(z)}$}\label{simplenu}

Let us construct the Separation of variables basis \cite{Nic13} of the modified creation operator $\nu_{12}(z)$. We assume that the inhomogeneity parameters $\bar \theta$ are generic complex numbers and that the product $AB$ has non zero entries.

Let us introduce a vector
\begin{gather*}
|y\>=\cA^{-1}_1\cA^{-1}_2\cdots \cA^{-1}_N\rvech,
\end{gather*}
where
\begin{gather*}
\cA_i=a_0 \exp{\big(a^+S_i^{+}+a^-S_i^{-}+2 a_3S_i^{3}\big)},
\end{gather*}
and the parameters $\{a_0,a^+,a^-,a_3\}$ are fixed by the equality
\begin{gather*}
A=\sqrt\mu\left(\begin{matrix} 1&\dfrac{\rho_2}{\km}\\ \dfrac{\rho_1}{\kp}&1\end{matrix}\right)=a_0 \exp{\big(a^+\sigma^{+}+a^-\sigma^{-}+a_3\sigma^{3}\big)}.
\end{gather*}
Using $\mathfrak{gl}_2$ invariance we can show that
\begin{gather*}
\nu_{12}(z)|y\>=\frac{\mu}{\km}(\rho_1 + \rho_2)\lambda_2(z)|y\>.
\end{gather*}
Indeed, we have
\begin{gather*}
\nu_{12}(z)|y\> = \Tr(E_{21}A T(z) B)\cA^{-1}_1\cA^{-1}_2\cdots \cA^{-1}_N\rvech
 = \cA^{-1}_1\cA^{-1}_2\cdots \cA^{-1}_N\Tr(E_{21}T(z)AB)\rvech\\
\hphantom{\nu_{12}(z)|y\>}{} =(AB)_{12}\lambda_2(z)|y\>.
\end{gather*}
Consider a set of vectors
\begin{gather*}
 |Y(\bar k)\>=\prod_{i=1}^N\prod_{j=1}^{k_i}\nu_{22}\left( \theta_i+c\left(s_i+\frac{1}{2}-j\right) \right)|y\>,
\end{gather*}
where $\bar k$ is a set of $N$ integers such that $k_i\in\{0,1,\dots, 2 s_i\}$ for $i=1,\dots,N$. Then it follows from the commutation relation~\eqref{comsl22a} that
\begin{gather*}
\nu_{12}(z)\prod_{j=1}^{n}\nu_{22}( u-c j )=\left(\prod_{j=1}^{n}f(z,u-c j)\right)\left(\prod_{j=1}^{n}\nu_{22}( u-c j )\right)\nu_{12}(z)\nonumber\\
\qquad {} + g(u-c,z)\left(\prod_{j=2}^{n}f(u-c,u-c j)\right)\nu_{22}(z)\left(\prod_{j=2}^{n}\nu_{22}( u-c j )\right)\nu_{12}(u-c).
\end{gather*}
 Here we have used $f(x-c,x)=0$ and $\lambda_2\big(\theta_i+c\big(s_i-\frac{1}{2}\big)\big)=0$. Hence, the action of $\nu_{12}(z)$ on~$|Y(\bar k)\>$ reads
\begin{gather*}
\nu_{12}(z)|Y(\bar k)\>=\frac{\mu}{\km}(\rho_1 + \rho_2)\frac{F(z)}{\Lambda_{12}(z,\bar k)}|Y(\bar k)\>,
\end{gather*}
where $F(z)$ given by (\ref{F-rep}) and
\begin{gather*}
\Lambda_{12}\big(z,\bar k\big)=\prod_{i=1}^{N} \left(\prod_{j=1}^{k_i} h(z-c,\theta_i+c\left(s_i+\frac{1}{2}-j\right)\right)\left(\prod_{j=k_i+1}^{2s_i} h(z,\theta_i+c\left(s_i+\frac{1}{2}-j\right)\right),
\end{gather*}
with $h(u,v)=f(u,v)/g(u,v)$.

\begin{rmk}Using $\mathfrak{gl}_2$ invariance we can show that
\begin{gather*} \big\langle \hat{0}|\cA^{-1}_1\cA^{-1}_2\cdots \cA^{-1}_N|Y(\bar k)\big\rangle =\prod_{i=1}^N\prod_{j=1}^{k_i}(AB)_{22}\lambda_1\left( \theta_i+c\left(s_i+\frac{1}{2}-j\right) \right)\neq 0.
\end{gather*}
 Thus, all the vectors under consideration are non zero.
\end{rmk}

Thus, we have constructed $\prod\limits_{i=1}^N(2s_i+1)$ different eigenvectors of the modified creation operator with different eigenvalues. As the vectors are independent and the representation is finite (being of dimension $\prod\limits_{i=1}^N(2s_i+1)$), this implies that the modified creation operator has simple spectrum and it is invertible. Thus, we have{\samepage
\begin{gather*}
F(z)(\nu_{12}(z))^{-1}|Y(\bar k)\>=\frac{\km}{\mu (\rho_1 + \rho_2)}\Lambda_{12}(z,\bar k)|Y(\bar k)\>.
\end{gather*}
Therefore, $F(z)(\nu_{12}(z))^{-1}$ is a polynomial in $z$ of degree $z^S$.}

Observe that the modified creation operator
\begin{gather*}
\nu_{12}(z)=\Tr(V_{12}T(z)), \qquad V_{12}=\mu\begin{pmatrix} \dfrac{\rho_1}{\km}& \dfrac{\rho_1\rho_2}{(\km)^2}\vspace{1mm}\\
1&\dfrac{\rho_2}{\km}\end{pmatrix}
\end{gather*}
is a null twisted transfer matrix, because $\Det(V_{12})=0$. Since we can consider $\rho_i$, $\mu$, and $\km$ as free parameters in our construction, the result applies to any null twisted transfer matrix.

\subsection*{Acknowledgements} S.B.\ thanks R.~Pimenta, N.~Cramp\'e, P.~Baseilhac, V.~Pasquier, D.~Serban, A.~Faribault for discussions. N.S.\ would like to thank the hospitality of the Institute de Physique Th\'{e}orique at
the CEA de Saclay where a part of this work was done. S.B.\ and B.V.\ would like to thank the hospitality of the LMPT of Tours where a part of this work was done. S.B.\ was supported by a~public grant as part of the
Investissement d'avenir project, reference ANR-11-LABX-0056-LMH, LabEx LMH. N.S.\ was supported by the Russian Foundation RFBR-18-01-00273a.

\pdfbookmark[1]{References}{ref}
\LastPageEnding

\end{document}